\documentclass[runningheads,envcountsame]{llncs}
\usepackage[T1]{fontenc}
\usepackage{soul}
\usepackage{graphicx}

\usepackage[bookmarks,unicode,colorlinks=true]{hyperref}%
   \def\@citecolor{blue}%
   \def\@urlcolor{blue}%
   \def\@linkcolor{blue}%
\def\orcidID#1{\smash{\href{http://orcid.org/#1}{\protect\raisebox{-1.25pt}{\protect\includegraphics{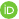}}}}}
\usepackage{color}

\usepackage{amsmath}

\usepackage{cleveref} 
\Crefname{figure}{Fig.}{Figs.}

\usepackage{xspace}
\usepackage{mathtools}
\usepackage{subcaption}
\usepackage{fontawesome5}
\usepackage{lineno}

\usepackage{tikz}
\usetikzlibrary{automata,positioning,decorations.markings,arrows,intersections,calc,shapes,patterns,patterns.meta,fit}
\usepackage{adjustbox}  

\usepackage{thmtools}
\usepackage{amsfonts}
\usepackage{todonotes}
\usepackage{enumitem}
\setlist[description]{leftmargin=0pt,labelindent=0pt,itemsep=2pt,topsep=2pt}

\newcommand{\Aa}{{\mathcal{A}}}

\newcommand{\Gg}{{\mathcal G}} 
 
\newcommand{\toss}{{\mathsf{t}}}

\newcommand{\Max}{\mathsf{Max}}
\newcommand{\Min}{\mathsf{Min}}
\newcommand{\PSpace}{$\mathbf{PSPACE}$}
\newcommand{\SharpP}{$\#\mathbf{P}$}
\newcommand{\UP}{$\mathbf{UP}$}
\newcommand{\NP}{$\mathbf{NP}$}
\newcommand{\NL}{$\mathbf{NL}$}

\newcommand{\classicRTG}{{random-turn~game}} 

\newcommand{\sgameone}{{unfolding~random-turn~game}} 
\newcommand{\SGAMEone}{{Unfolding~Random-Turn~Game}} 
\newcommand{\sgametwo}{{random~arena~game}} 
\newcommand{\Sgametwo}{{Random~arena~game}} 
\newcommand{\SGAMEtwo}{{Random~Arena~Game}} 

\newcommand{\own}{{\mathrm{own}}} 

\newcommand{\seq}[1]{\langle #1 \rangle} 

\newcommand{\Plays}{\mathsf{Plays}}
\newcommand{\incl}{{\subseteq}}

\newcommand{\modify}[1]{\textcolor{black}{#1}}

\begin{document}\sloppy
\title{The Complexity of Games with Randomised Control}
%
%
%
\author{
Sarvin Bahmani\orcidID{0009-0002-4476-489X} 
\and
Rasmus Ibsen-Jensen\orcidID{0000-0003-4783-0389} \and
Soumyajit Paul\orcidID{0000-0002-7233-2018} \and 
Sven Schewe\orcidID{0000-0002-9093-9518} \and 
Friedrich Slivovsky\orcidID{0000-0003-1784-2346} \and 
Qiyi Tang\orcidID{0000-0002-9265-3011} \and
Dominik Wojtczak\orcidID{0000-0001-5560-0546} \and
Shufang Zhu\orcidID{0000-0002-5922-8750}}

%

\institute{
University of Liverpool, Liverpool, UK\\
\email{\{r.bahmani, r.ibsen-jensen, soumyajit.paul, sven.schewe, f.slivovsky, qiyi.tang, dkw, shufang.zhu\}@liverpool.ac.uk}
}
\authorrunning{Bahmani et al.}
%
%
\maketitle              

\begin{abstract}
We study the complexity of solving 
two-player infinite duration games played on a fixed finite graph, where the control of a node 
is not predetermined but rather assigned randomly.  
In classic \emph{random-turn games}, control of each node is assigned randomly every time the node is visited during a play. In this work, we study two natural variants of this where control of each node is assigned only once:   
(i) control is assigned randomly during a play when a node is visited for the first time and does not change for the rest of the play and (ii) control is assigned a priori before the game starts for every node by independent coin tosses and then the game is played. 

We investigate the complexity of computing the winning probability with three kinds of objectives—reachability, parity, and energy. We show that the qualitative questions on all variants and all objectives are \NL-complete. For the quantitative questions, we show that deciding whether the maximiser can win with probability at least a given threshold for every objective is \PSpace-complete under the first mechanism, 
and that computing the exact winning probability for every objective is \SharpP-complete under the second.  
To complement our hardness results for the second mechanism, we propose randomised approximation schemes that efficiently estimate the winning probability \modify{for all three objectives, assuming a bounded number of parity colours and unary-encoded weights for energy objectives,} and we empirically demonstrate their fast convergence.

\end{abstract}
\section{Introduction}\label{sec:intro}

Graph games are fundamental framework for modelling and studying decision making scenarios in formal verification~\cite{EmersonJ91,kupferman1997module,DeAlfaro2001279,alur2002alternating} and reactive synthesis~\cite{PR89,Bloem2018}.  In the most classic setting, these are used to capture a system interacting with an adversarial environment. This can be naturally perceived as a two-player game between the system and the environment where the system’s desired requirements are formulated as a winning condition of the game. 
Classic winning conditions on infinite traces include reachability, parity and energy~\cite{fijalkow2023gamesgraphs}. 
In two-player turn-based games on graphs, a token is placed on an initial node of the graph, and the players take turns moving it along the edges, thereby generating an infinite trace.
Each node is typically owned by one of the players: the system selects the successor node on its own nodes, while the environment does so on its respective ones.
A standard assumption in these two-player games on graphs is that the ownership of nodes is fixed before the game begins.
Each node belongs permanently to one of the players, and control never changes hands throughout the game. But this assumption disregards situations where there is some uncertainty in who controls a decision. Randomness being a natural simulator of uncertainty, the controls of a node can be decided with a coin toss, giving rise to what are known as \emph{random-turn games}. 

Random-turn games have been studied for combinatorial games~\cite{Peres2005RandomTurnHA} as well as for infinite duration graph games~\cite{AHI18}. These games have also found applications in bargaining networks~\cite{Celis-et-al-2010}. In random-turn games, every time a node is visited, the ownership of this node is decided randomly and independently. While this mechanism is simple, it also grants the players an unbounded amount of leeway to make non-optimal choices in the initial stages of the game. 
Moreover, this does not capture scenarios where the control of a decision point depends only on some single uncertain event. In this work, we study random-turn games where ownership of a particular node is decided randomly but only once, i.e., the ownership of a node is fixed once it is assigned.
This is already the case in random-turn games for some board games with acyclic state graph as studied in~\cite{Peres2005RandomTurnHA}.

We study two natural variants of this: (i) the ownership of nodes are decided gradually on the fly as the game goes on and fixed for the rest of the game, termed as \emph{\SGAMEone{}}--these games can be perceived as the unfolding of a larger stochastic game; (ii) the ownership of all nodes are decided a priori randomly and independently before the start of the game, termed as \emph{\SGAMEtwo{}}--this is equivalent to randomly selecting a \emph{game arena} for the game.
With randomness in the picture, the fundamental question that we are concerned with is computing the probability that the system player wins. Random-turn games, \sgameone{} and \sgametwo{}, all of them can be reduced to equivalent stochastic games~\cite{Condon92}. 
For random-turn games, this translation is actually in deterministic logarithmic space, resulting in a stochastic game of polynomial size. 
However, it remains unknown whether any stochastic game can be reduced to a structure-preserving (i.e., a simple translation in deterministic logarithmic space), probability-preserving random-turn game.
On the other hand, for \sgameone{} and \sgametwo{} the equivalent stochastic games can suffer from an exponential blow-up in size. In this paper, we consider three classic objectives: reachability, parity and energy. \Cref{exp:games} illustrates these different variants of games with randomised control on a simple graph in \Cref{fig:explicit-game}(a) with a reachability objective.



\begin{example}\label{exp:games}

\begin{figure}[t]

\scalebox{.78}{
\begin{tikzpicture}[scale=1, >=latex', shorten >=2pt, node distance=.7cm and 1.2cm, on grid, auto]
    \begin{scope}[shift={(-3,0)}]
    \node[state, initial, initial text=] (n1) {$v_0$};
    \node[state, above right=of n1] (n2) {$v_1$};
    \node[state, below right=of n1] (n3) {$v_2$};
    \node[state, right=2.7cm of n1] (n4) {$\top$};

    \path[->] (n1) edge[bend left] (n2);
    \path[->] (n2) edge[bend left] (n1);
    \path[->] (n1) edge[bend left] (n3);
    \path[->] (n3) edge[bend left] (n1);
    \path[->] (n2) edge (n4);
    \path[->] (n3) edge (n4);
    \path[->] (n4) edge[loop above] (n4);

    \node[below right=1cm and 0cm of n3] {\parbox{6cm}{(a) A graph with reachability objective of reaching $\top$.}};
    \end{scope}
    
    \begin{scope}[scale=0.8, >=latex', shorten >=1pt, node distance=0.45cm and 1.3cm, on grid, auto, shift={(3,0)}]
    \node[state, dashed, initial, initial text=] (v0) {$v_0$};
    \node[state, above right=of v0, fill=blue!20] (v0+) {$v_0^+$};
    \node[state, below right=of v0, pattern={Lines[angle=0,distance=3pt,line width=0.3pt]}, pattern color=red!50] (v0-) {$v_0^-$};
    
    \node[state, dashed,right=2.6cm of v0] (v1) {$v_1$};
    \node[state, above right=of v1, fill=blue!20] (v1+) {$v_1^+$};
    \node[state, below right=of v1, pattern={Lines[angle=0,distance=3pt,line width=0.3pt]}, pattern color=red!50] (v1-) {$v_1^-$};  
    
    \node[state, dashed, right=2.6cm of v1] (v2) {$v_2$};
    \node[state, above right=of v2, fill=blue!20] (v2+) {$v_2^+$};
    \node[state, below right=of v2, pattern={Lines[angle=0,distance=3pt,line width=0.3pt]}, pattern color=red!50] (v2-) {$v_2^-$};     
    
    \node[state, right=2.6cm of v2] (t) {$\top$};

    \path[->] (v0) edge (v0+);
    \path[->] (v0) edge (v0-);

    \path[->] (v0+) edge (v1);
    \path[->] (v0-) edge (v1);
    \path[->] (v0+) edge [bend left=39] (v2);
    \path[->] (v0-) edge [bend left=-39] (v2);
  
    \path[->] (v1) edge (v1+);
    \path[->] (v1) edge (v1-);
    \path[->] (v1+) edge [bend right=39] (v0);
    \path[->] (v1+) edge [bend left=39] (t);
    
    \path[->] (v1-) edge [bend right=-39] (v0);
    \path[->] (v1-) edge [bend left=-39] (t);

    \path[->] (v2) edge (v2+);
    \path[->] (v2) edge (v2-);
    \path[->] (v2+) edge (t);
    \path[->] (v2-) edge (t);
    
    \path[->] (v2+) edge [bend right=40] (v0);
    \path[->] (v2-) edge [bend right=-40] (v0);

    \path[->] (t) edge [loop above] (t);
    
    \node[below right=1.9cm and 1cm of v1] {\parbox{8cm}{(b) Stochastic game reduced from the {\classicRTG}.}}; 

    \end{scope}
    
    \begin{scope}[shift={(-2,-3.5)}]   
    \node[state] (n1) {$v_0$};
    \node[state, above right=of n1] (n2) {$v_1$};
    \node[state, below right=of n1] (n3) {$v_2$};
    \node[state, right=2.4cm of n1] (n4) {$\top$};

    \path[->] (n1) edge[bend left] (n2);
    \path[->] (n2) edge[bend left] (n1);
    \path[->] (n1) edge[bend left] (n3);
    \path[->] (n3) edge[bend left] (n1);
    \path[->] (n2) edge (n4);
    \path[->] (n3) edge (n4);
    \path[->] (n4) edge[loop above] (n4);
    \fill[black] ([yshift=2mm] n1.north) circle (2pt);
    
    \node[state, below left =2.7cm and 3cm of n3, fill=blue!20] (0n1) {$v_0$};
    \node[state, above right=of 0n1] (0n2) {$v_1$};
    \node[state, below right=of 0n1] (0n3) {$v_2$};
    \node[state, right=2.4cm of 0n1] (0n4) {$\top$};

    \path[->] (0n1) edge[bend left] (0n2);
    \path[->] (0n2) edge[bend left] (0n1);
    \path[->] (0n1) edge[bend left] (0n3);
    \path[->] (0n3) edge[bend left] (0n1);
    \path[->] (0n2) edge (0n4);
    \path[->] (0n3) edge (0n4);
    \path[->] (0n4) edge[loop above] (0n4);
    
    \fill[black] ([yshift=2mm] 0n1.north) circle (2pt);

    \node[state, below right =2.7cm and 1cm of n3, pattern={Lines[angle=0,distance=3pt,line width=0.3pt]}, pattern color=red!50] (1n1) {$v_0$};
    \node[state, above right=of 1n1] (1n2) {$v_1$};
    \node[state, below right=of 1n1] (1n3) {$v_2$};
    \node[state, right=2.4cm of 1n1] (1n4) {$\top$};

    \path[->] (1n1) edge[bend left] (1n2);
    \path[->] (1n2) edge[bend left] (1n1);
    \path[->] (1n1) edge[bend left] (1n3);
    \path[->] (1n3) edge[bend left] (1n1);
    \path[->] (1n2) edge (1n4);
    \path[->] (1n3) edge (1n4);
    \path[->] (1n4) edge[loop above] (1n4);
    
    \fill[black] ([yshift=2mm] 1n1.north) circle (2pt);
    \node[draw, rounded corners, thick, fit=(n1)(n2)(n3)(n4),dashed] (boxTop) {};
    \node[draw, rounded corners, thick, fit=(0n1)(0n2)(0n3)(0n4),] (boxLeft) {};
    \node[draw, rounded corners, thick, fit=(1n1)(1n2)(1n3)(1n4),](boxRight) {};

    \node[below left=2cm and 1cm of boxLeft] (voidA) {};
    \node[below right=2cm and 1cm of boxLeft] (voidB) {};
    \node[below left=2cm and 1cm of boxRight] (voidC) {};
    \node[below right=2cm and 1cm of boxRight] (voidD) {};


    \draw[->, very thick] (boxTop.south) -- node[midway, left=5mm, font=\small]{$\frac{1}{2}$} (boxLeft.north);
    \draw[->, very thick] (boxTop.south) -- node[midway, right=5mm, font=\small]{$\frac{1}{2}$} (boxRight.north);
    
    \draw[->, very thick] (boxLeft.south) -- node[midway, left=2mm, font=\small]{$v_1$} (voidA);
    \draw[->, very thick] (boxLeft.south) -- node[midway, right=2mm, font=\small]{$v_2$} (voidB);

    \draw[->, very thick] (boxRight.south) -- node[midway, left=2mm, font=\small]{$v_1$} (voidC);
    \draw[->, very thick] (boxRight.south) -- node[midway, right=2mm, font=\small]{$v_2$} (voidD);
    
    \node[below=6cm of boxTop] {\parbox{8cm}{(c) Stochastic (explicit) game reduced from the {\sgameone} -- the first two levels.}};
    \end{scope}
    
    \begin{scope}[shift={(4,-4.2)}]
    \node[state, fill=blue!20] (0n1) {$v_0$};
    \node[state, fill=blue!20, above right=of 0n1] (0n2) {$v_1$};
    \node[state, below right=of 0n1, pattern={Lines[angle=0,distance=3pt,line width=0.3pt]}, pattern color=red!50] (0n3) {$v_2$};
    \node[state, right=2.4cm of 0n1] (0n4) {$\top$};

    \path[->] (0n1) edge[bend left] (0n2);
    \path[->] (0n2) edge[bend left] (0n1);
    \path[->] (0n1) edge[bend left] (0n3);
    \path[->] (0n3) edge[bend left] (0n1);
    \path[->] (0n2) edge (0n4);
    \path[->] (0n3) edge (0n4);
    \path[->] (0n4) edge[loop above] (0n4);

    \node[state, right =3.5cm of 0n1, fill=blue!20] (1n1) {$v_0$};
    \node[state, above right=of 1n1, pattern={Lines[angle=0,distance=3pt,line width=0.3pt]}, pattern color=red!50] (1n2) {$v_1$};
    \node[state, fill=blue!20, below right=of 1n1] (1n3) {$v_2$};
    \node[state, right=2.4cm of 1n1] (1n4) {$\top$};

    \path[->] (1n1) edge[bend left] (1n2);
    \path[->] (1n2) edge[bend left] (1n1);
    \path[->] (1n1) edge[bend left] (1n3);
    \path[->] (1n3) edge[bend left] (1n1);
    \path[->] (1n2) edge (1n4);
    \path[->] (1n3) edge (1n4);
    \path[->] (1n4) edge[loop above] (1n4);
    
    \node[draw, rounded corners, thick, fit=(0n1)(0n2)(0n3)(0n4),] (boxLeft) {};
    \node[draw, rounded corners, thick, fit=(1n1)(1n2)(1n3)(1n4),](boxRight) {};

    \node[state, below=2.7cm of 0n1, fill=blue!20] (2n1) {$v_0$};
    \node[state, above right=of 2n1, fill=blue!20] (2n2) {$v_1$};
    \node[state, below right=of 2n1, fill=blue!20] (2n3) {$v_2$};
    \node[state, right=2.4cm of 2n1] (2n4) {$\top$};

    \path[->] (2n1) edge[bend left] (2n2);
    \path[->] (2n2) edge[bend left] (2n1);
    \path[->] (2n1) edge[bend left] (2n3);
    \path[->] (2n3) edge[bend left] (2n1);
    \path[->] (2n2) edge (2n4);
    \path[->] (2n3) edge (2n4);
    \path[->] (2n4) edge[loop above] (2n4);

    \node[state, right =3.5cm of 2n1, pattern={Lines[angle=0,distance=3pt,line width=0.3pt]}, pattern color=red!50] (3n1) {$v_0$};
    \node[state, above right=of 3n1, fill=blue!20] (3n2) {$v_1$};
    \node[state, below right=of 3n1, fill=blue!20] (3n3) {$v_2$};
    \node[state, right=2.4cm of 3n1] (3n4) {$\top$};

    \path[->] (3n1) edge[bend left] (3n2);
    \path[->] (3n2) edge[bend left] (3n1);
    \path[->] (3n1) edge[bend left] (3n3);
    \path[->] (3n3) edge[bend left] (3n1);
    \path[->] (3n2) edge (3n4);
    \path[->] (3n3) edge (3n4);
    \path[->] (3n4) edge[loop above] (3n4);
    
    \node[draw, rounded corners, thick, fit=(2n1)(2n2)(2n3)(2n4)] (boxLeftBelow) {};
    \node[draw, rounded corners, thick, fit=(3n1)(3n2)(3n3)(3n4),](boxRightBelow) {};
    
    \node[below left=5.3cm and 0.8cm of 1n1] {\parbox{5.2cm}{(d) All four winning arenas for the {\sgametwo}.}};
    \end{scope}
    
\end{tikzpicture}
}
\caption{
Blue nodes are controlled by $\Max$, pink nodes with dashed patterns by $\Min$, white nodes with dashed contours are random, and white nodes with solid contours are unassigned. 
In (c) and (d), boxes denote game states.
Black dots above some nodes indicate token positions.
Transition probabilities from each random node are uniformly distributed.
}
\label{fig:explicit-game}
\end{figure}
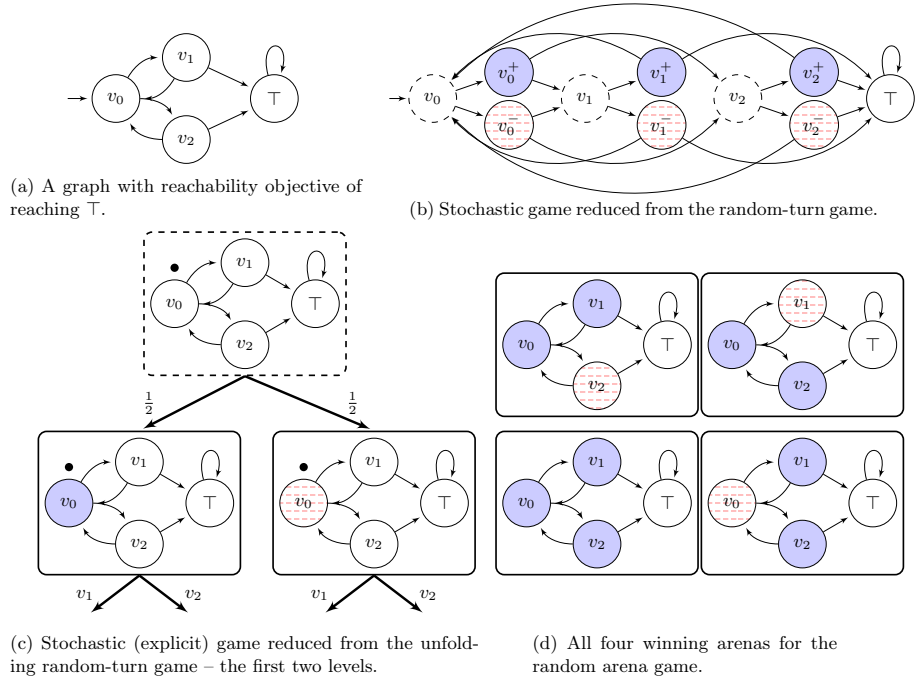

In the game graph in \Cref{fig:explicit-game}(a), the token is initially on state $s_0$. The players are $\Max$ and $\Min$, and $\Max$'s objective is to maximise the probability of reaching the target node $\top$. Let us consider the 
games with randomised control
on this graph where ownership is always determined uniformly randomly. First, in the \sgametwo{}, each node from $\{v_0,v_1,v_2\}$ is assigned to $\Max$ or $\Min$ with probability $\frac{1}{2}$. Hence, each of the $2^3 =8$ possible \emph{game arenas} is obtained with probability $\frac{1}{8}$. Among these, $\Max$ has a winning strategy in 4 of them, with ownership among one of $\{v_0,v_1\},\{v_0,v_2\},\{v_0,v_1,v_2\},\{v_1,v_2\}$; see \Cref{fig:explicit-game}(d). Hence, $\Max$ wins \sgametwo{} with probability $\frac{1}{2}$ on this graph. 

On the other hand, in the \sgameone{} on this graph, the initial node $v_0$ is assigned to $\Max$ or $\Min$ with probability $\frac{1}{2}$.
After the assignment, the respective player chooses a node among $v_1$ or $v_2$. By symmetry, since both options are equally good, suppose w.l.o.g.\ they choose $v_1$.
Next, $v_1$ is assigned randomly. If $v_0$ was assigned to $\Min$, then $\Max$ loses when $\Min$ gets $v_1$, as $\Min$ can play the cycle $v_0, v_1, v_0$. $\Max$ wins in the other case, when $\Max$ gets $v_1$.
If $v_0$ was assigned to $\Max$, then $\Max$ wins when $\Max$ gets $v_1$, whereas if $\Min$ gets $v_1$, they would move the token back to $v_0$. Since $\Max$ already owns $v_0$, they will then move to the still unassigned $v_2$. In the random assignment of $v_2$, $\Max$ wins if $v_2$ is assigned to $\Max$ and loses otherwise.
Consequently, the probability of $\Max$ winning is $\frac{1}{2}(\frac{1}{2}+\frac{3}{4}) = \frac{5}{8}$. 

Moreover, for the \classicRTG{} on this graph, $\Max$ wins almost surely, that is, with probability one, since $\Max$ wins whenever they obtain $v_1$ or $v_2$, which occurs with probability one. 
\qed
\end{example}

\paragraph*{\bf Our Contributions.}


We begin by considering the problem of whether there exists a deterministic logarithmic-space translation from a stochastic game to a random-turn game that preserves the winning probability. 
We show that this is unlikely: first, we prove that the qualitative question for the random-turn game with a reachability objective is \NL-complete; then, we infer that a simple structure- and probability-preserving translation would imply $\mathbf{NL} = \mathbf{P}$.

The main focus of this paper is on two new and distinct ways of assigning control: (i) ownership is decided dynamically when a node is visited for the first time, and (ii) ownership is determined a priori by independent coin tosses for each node. 
These models differ from standard stochastic games, and we show that they exhibit markedly different computational complexities: the quantitative problems are \PSpace-complete for (i) and \SharpP-complete for (ii), while both retain \NL-completeness for the qualitative questions. 
These complexity results hold for safety, reachability, parity, and energy objectives.
For (ii), we complement the \SharpP-completeness result with a randomised approximation scheme that can efficiently estimate the probability of winning \modify{for all three objectives, assuming a bounded number of parity colours and unary-encoded weights for energy objectives}. 
Missing proofs can be found in the full version.

\paragraph*{\bf Related Work.}

Previous work has investigated \emph{random-turn games}~\cite{Peres09}, where ownership of states is (re-)assigned randomly each time a state is visited.  
Such games can be transformed, in deterministic logarithmic space, into \emph{stochastic games}~\cite{Condon92} with the same winning probabilities, and are therefore often regarded as a subclass of stochastic games.  
Since stochastic games with the objectives we consider can be solved in \NP~$\mathrel{\cap}$~co-\NP~\cite{Condon92,ZwickP96,ChatterjeeJH04StochasticParity,ChatterjeeH08StochasticMeanPayoff}, the same complexity bound applies to random-turn games.  
It has been shown that random-turn games with out-degree~2 can be solved in polynomial time~\cite{AHI18}, implying that there is no simple degree-preserving reduction in the opposite direction, as \emph{simple stochastic games}---that is, stochastic games where all random vertices have out-degree~2---are already as hard as general stochastic games~\cite{ZwickP96}.  

A rich body of work also explores \emph{graph games} in which control is decided through alternative mechanisms.  
In \emph{bidding games}, introduced by~\cite{lazarus1999combinatorial,lazarus1996richman}, each turn involves an ``auction'' between the two players to determine who moves the token.  
Different bidding rules have been studied, including \emph{Richman bidding}~\cite{avni2019infinite} (named after David Richman), where the higher bidder pays the lower bidder, and \emph{poorman bidding}~\cite{AHI18}, where the higher bidder pays the bank and the payment is lost.  
It was observed in~\cite{lazarus1999combinatorial} that \emph{reachability} Richman-bidding games (and only those) are equivalent to random-turn games.  

Another related model is that of \emph{pawn games}~\cite{AvniGG25pawns}, where control of vertices also evolves dynamically during play.  
Here, control is determined by ``pawns'': each pawn owns a subset of vertices in the game graph, and pawns are distributed among players, with ownership possibly changing throughout the game.  
As noted in the seminal work on pawn games~\cite{AvniGG25pawns}, these games differ fundamentally from bidding games: pawn games support richer and more flexible mechanisms for transferring control, whereas bidding games rely solely on strict auction-based mechanisms.  

\section{Preliminaries}\label{sec:prelim}

\noindent{\bf Game Graph and Arena.} 
In this work, we study the turn-based zero-sum games that are played between two players ---$\Max$ and $\Min$---on finite directed graphs, which are called \emph{game arenas}.
A finite directed graph is $G = (V,E)$, where $V$ is a non-empty set of nodes and $E \mathrel{\incl} V \times V$ is a set of non-empty edges. Since our games are potentially infinite duration games, we only consider graphs with no dead ends, i.e., for each $u \in V$, for some $v$ in $V$, $(u,v) \in E$.

A game \emph{arena} on a graph $G = (V, E)$ is a tuple $\Aa = (G,V_0,V_1,v_0)$ where $V$ is the set of nodes, $E \subseteq V \times V$ is the set of edges, and $v_0$ is the start node.
The set of nodes $V$ is partitioned into a set $V_0$ of nodes controlled by player $\Max$
and a set $V_1$ of nodes controlled by player $\Min$. 
A play of such a game starts by placing a token on the start node $v_0$.
The player controlling this node then chooses a successor node
$v_1$ such that $(v_0, v_1) \in E$ and the token is moved to this successor node. 
In the next turn the player controlling the node $v_1$ chooses the successor node $v_2$ with $(v_1, v_2) \in E$ and the token is moved accordingly. 
Both players move the token over the arena in this manner and thus form a play of the game.
Note that, on a graph $G$ with node set $V$, the total number of possible arenas is $2^{|V|}$.

\noindent{\bf Play and Strategies.}
Formally, a \emph{play} of a game $\Aa$ is a finite (resp. an infinite) sequence of nodes $\seq{v_0, v_1, \ldots} \in V^*$ (resp. $\seq{v_0, v_1, \ldots} \in V^\omega$) such that, for all $i \geq 0$, we have that $(v_i, v_{i+1}) \in E$.  
We write $\Plays_\Aa(v)$ for the set of plays of the game $\Aa$ that start from a node $v \in V$ and $\Plays_\Aa$ for the set of plays of the game. 
We omit the subscript when the arena is clear from the context. 

A \emph{strategy} for player $\Max$ on a given game arena, is a function $\sigma: V^*V_0 \rightarrow V$ such that $\big(v,\sigma(\rho,v)\big)\in E$ for all $\rho \in V^*$ and $v \in V_0$.
A strategy $\sigma$ is called memoryless if $\sigma$ only depends on the last state ($\sigma(\rho,v) = \sigma(\rho',v)$ for all $\rho,\rho' \in V^*$ and $v \in V_0$). 
A play $\seq{v_0, v_1, \ldots}$ is consistent with $\sigma$ if, for every initial sequence $\rho_n = v_0, v_1, \ldots,v_n$ of the play that ends in a state of player $\Max$ ($v_n \in V_0$), $\sigma(\rho_n)=v_{n+1}$ holds.
Strategies for player $\Min$ are defined analogously. 

\noindent\modify{
{\bf Intuitive Notion of Value.}
Given fixed strategies for both players, that is, a strategy profile, the outcome of the game is an infinite play, which induces a probability that the objective of player $\Max$ is satisfied. We refer to this probability as the value of the strategy profile. The value of a strategy for a player is then defined as the probability of winning against an optimal strategy of the opponent. Finally, the value of the game (from a given initial vertex) is the winning probability that player $\Max$ can guarantee when both players play optimally. Intuitively, this value captures how likely player $\Max$ is to achieve the objective under optimal play.}

\noindent{\bf Game Objectives.}
We consider three objectives: reachability, parity and energy.
For \emph{reachability} games, we have a target node $\top \in V$.
If $\top$ is visited during a play, player $\Max$ is declared the winner and the game immediately terminates.
Otherwise, if $\top$ is never visited, player $\Min$ is the winner.

For \emph{parity} objectives, we equip the game with a priority function $\mathrm{pri}: V \to \mathbb{N}$. 
A play $\seq{v_0, v_1, \ldots}$ is won by player $\Max$ if the maximal priority seen infinitely often is even, that is, $\limsup_{i \rightarrow \infty}\mathrm{pri}(v_i)$ is even, and by player $\Min$ otherwise. 

For \emph{energy} objectives, the game has an initial credit $c \in \mathbb{N}$ and a weight function $w: V \to \mathbb{Z}$.
The goal of player $\Max$ is to construct an infinite play $\seq{v_0, v_1, \ldots}$ such that 
\begin{equation}\label{eq:energy-nonnegative}
    c + \sum_{i=0}^{j} w(v_i) \ge 0 \; \text{for all} \; j \ge 0.
\end{equation}
The quantity $c + \sum_{i=0}^{j} w(v_i)$ is called the \emph{energy level} of the play prefix $\seq{v_0, v_1, \ldots, v_j}$. 
A play $\seq{v_0, v_1, \ldots}$ is won by player $\Max$ if it satisfies \Cref{eq:energy-nonnegative}, otherwise it is won by player $\Min$. 

A node $v \in V$ is winning for player $\Max$ (resp. $\Min$) if there exists a winning strategy for player $\Max$ (resp. $\Min$) from $v$. 
All three games are memoryless determined \cite{Ehrenfeucht1979,Gurvich88,mostowski1991games,EmersonJ91,BouyerFLMS08}; that is, for all $v \in V$, $v$ is winning for either player $\Max$ or player $\Min$, and memoryless strategies are sufficient.

\noindent{\bf Games with Randomised Control.}
In a game with randomised control, we play on a fixed game graph $G$ with a designated starting node $v_0$, where the objective may be reachability, parity, or energy.
Although the graph $G$ is fixed, the control of its nodes is not predetermined; initially there is no partition of $V$ into $V_0$ and $V_1$.
Instead, the control of each node $v \in V$ is determined randomly and independently by a coin toss: 
given a function $\toss: V \to (0,1)$, the node $v$ is assigned to player $\Max$ with probability $\toss(v)$~(i.e., $v \in V_0$) and to player $\Min$ with probability $1 - \toss(v)$ (i.e., $v \in V_1$).
Thus, the partition $(V_0,V_1)$ is determined entirely by these coin tosses.

Let $\own: V \rightharpoonup \{\Max, \Min\}$ be a partial function that, when defined, assigns a player to a node of the graph. 
The game starts with an empty ownership function, $\own = \emptyset$.
The ownership of the nodes is determined randomly during gameplay.  
We consider two variations of the classic random-turn games~\cite{Peres09}, 
in which the player who moves in each round is determined by a coin toss. 
In an \emph{{\sgameone}}, the ownership of a node is decided upon its first visit and remains fixed thereafter. 
In a \emph{{\sgametwo}}, the ownership of all nodes is determined at the start of the game.

\noindent{\bf Stochastic Game.}
A stochastic game is also a turn-based game played between players $\Max$ and $\Min$ on a game arena.
In contrast to non-stochastic games, stochastic games include random nodes.
Formally, the set of nodes $V$ is partitioned into three disjoint subsets: $V_0$, controlled by $\Max$; $V_1$, controlled by $\Min$; and $V_r$, representing random nodes.
Transitions from random nodes are probabilistic, specified by a transition function $\delta : V_r \to \mathrm{Dist}(V)$, where $\mathrm{Dist}(V)$ denotes the set of probability distributions over $V$.
A \classicRTG{} can be viewed as a special case of a stochastic game.
An \sgameone{}, as we will show later, can be transformed into an exponentially large stochastic game, called the explicit game, while \sgametwo{} can be transformed into a trivial stochastic game with an initial random node that transitions to game arenas with different control assignments, where the transition probability corresponds to the likelihood of generating each random assignment.


    


\section{Qualitative Analysis and Random-Turn Games}\label{sec:qualitative}

In this section, we discuss the qualitative analysis of games with randomised control, focusing on almost-sure and sure winning for $\Max$ in {\classicRTG}s, {\sgameone}s, and {\sgametwo}s.
A game is almost-surely won by player $\Max$ if it is won with probability one, and surely won if $\Max$ wins regardless of the outcomes of the coin tosses.
We show that it is \NL-complete to decide whether $\Max$ wins surely or almost-surely for {\classicRTG}s with a reachability objective, and for {\sgameone}s and {\sgametwo}s with reachability, parity, or energy objectives.
We also discuss reductions between {\classicRTG}s and stochastic games.

We first show deciding whether $\Max$ almost surely wins a random-turn game with a reachability objective is \NL-complete. 
The reduction is from st-connectivity, a well-known \NL-complete problem.
\begin{restatable}{theorem}{thmAlmostRTG}\label{thm:almost-sure-rtg}
    It is \NL-complete to decide whether $\Max$ almost surely wins a random-turn game with a reachability objective.
\end{restatable}

For random-turn games, almost-sure winning and sure winning differ in general — for instance, $\Max$ wins the random-turn game played on \Cref{fig:explicit-game}(a) almost surely but not surely, see the equivalent stochastic game in \Cref{fig:explicit-game}(b): 
if it happens that every time $v_1$ and $v_2$ are assigned to $\Min$ when visited, then $\Max$ loses; however, this event occurs with probability zero.
Nevertheless, deciding whether $\Max$ surely wins is also \NL-complete.


This result extends to parity and energy objectives: for parity, $\Max$ does not win surely iff there is a ``bad'' lasso-path whose cycle has an odd maximum priority; for energy, $\Max$ does not win surely iff there is a ``bad'' lasso-path where the energy level drops below zero anywhere on this path or the cycle is negative. 
In both cases, checking the existence of such a lasso-path can be done in \NL.

\begin{restatable}{theorem}{thmSureRTG}\label{thm:sure-rtg}
    It is \NL-complete to decide whether $\Max$ surely wins a random-turn game with a reachability, parity, or energy objective.
\end{restatable}

As a consequence of the \NL-completeness of the qualitative solutions to random-turn games, there is no simple structure%
\footnote{With structure preserving we mean any simple translation, such as gadgets, that can be calculated in deterministic logarithmic space.} and probability-preserving translation from stochastic games to random-turn games, even when the stochastic games are restricted to non-stochastic two-player games.

\begin{theorem}
If $\mathbf{NL}\neq\mathbf{P}$, then there is no structure- and probability-preserving translation from two-player reachability games to random-turn games.
\end{theorem}

\begin{proof}
Solving two-player reachability games is $\mathbf{P}$-complete \cite{Goldschlager1977}.
A deterministic logarithmic space (logspace) reduction that guarantees probability preservation to an \NL-complete problem would therefore imply that we can solve this reachability problem in \NL, and thus $\mathbf{NL} = \mathbf{P}$.
\end{proof}

While deterministic logspace reductions sound unusual for more general games, we note that the reductions from parity to mean-payoff and from mean-payoff to discounted-payoff in \cite{Jurdzinski98} are deterministic  logspace reductions; as $\mathbf{P}$-hardness is not an issue for these games, the inclusion in $\mathbf{L}$ is usually not emphasised.
The argument of non-reducibility extends to parity and energy games, as they can encode reachability with 2 priorities and weights $0$ and $-1$, respectively.

\begin{corollary}
If $\mathbf{NL}\neq\mathbf{P}$, then there is no structure- and probability-preserving translation from two-player parity or energy games to random-turn games.
\end{corollary}

We now consider almost-sure and sure winning for {\sgameone}s and {\sgametwo}s.
We first observe that sure winning for {\sgameone} and {\sgametwo} is the same as sure winning for random-turn games.
The next thing we observe is that almost-sure winning coincides with sure winning for {\sgameone}s and {\sgametwo}s.
By \Cref{thm:sure-rtg}, we immediately have:
\begin{restatable}{theorem}{thmSureAllGames}\label{thm:sure-almost-sure-all-games}
    It is \NL-complete to decide whether $\Max$ wins an {\sgameone} or a {\sgametwo} almost surely or surely, for reachability, parity, or energy objectives.
\end{restatable}

\section{\SGAMEone}\label{sec:toss-as-you-go}
In this section, we show that deciding whether the winning probability of $\Max$ for an \sgameone\ is at least a given threshold $\theta$, with any of the objectives—reachability, parity, or energy—is \PSpace-complete.

We state the main complexity results below;  the upper bound is proved in \Cref{subsec:in-pspace} and the lower bound in \Cref{subsec:pspace-hardness}, respectively. 

\begin{theorem}\label{thm:pspace-complete}
Checking whether player $\Max$ can win the \sgameone\ with a given reachability, parity, or energy objective with a probability at least 
\modify{a given rational threshold~$\theta$, encoded in binary,}
is \PSpace-complete.
\end{theorem}

\subsection{Membership in PSPACE}\label{subsec:in-pspace}
To determine the winning probability of player~$\Max$ in {\sgameone}s, 
one can transform the game into an equivalent stochastic reachability, parity, or energy game, 
called the \emph{explicit game}, and then solve this explicit game directly.  
We formally define the explicit game as follows.

\paragraph{Explicit Game} %
Let $G = (V, E)$ be the game graph of an {\sgameone}, and let $\toss$ be the associated coin-toss function. 
\modify{We assume that $\toss$ assigns rational probabilities to vertices, each given by a finite binary encoding.} 
The \emph{explicit game} $\Gg$ is a stochastic game whose states are of the form $(G, \own, u)$, 
where $\own$ is an ownership function and $u \in V$ is the current token position.  
The initial state is $(G, \emptyset, v_0)$, where no ownership is assigned to any node and $v_0 \in V$ is the initial node in the original {\sgameone}.  

The transitions in the explicit game are defined as follows. 
At a game state $(G, \own, u)$,
\begin{itemize}
    \item 
    if $\own(u)$ is undefined, then the state is random. It transitions with probability $\toss(u)$ to $(G, \own', u)$ where $\own' = \own \cup \{u \mapsto \Max\}$, 
    and with probability $1 - \toss(u)$ to $(G, \own'', u)$ where $\own'' = \own \cup \{u \mapsto \Min\}$;
    \item 
    otherwise, $\own(u)$ is defined, $(G, \own', u)$ is owned by player $\own(u)$ and has successors $(G, \own, v)$ for all $(u, v) \in E$.
\end{itemize}

The first two levels of the explicit game constructed from the \sgameone{} on the game graph in \Cref{fig:explicit-game}(a) is shown in \Cref{fig:explicit-game}(c).
A play $\langle (G, \emptyset, v_0), (G, \own_1, v_1), (G, \own_2, v_2), \ldots \rangle$ 
is won by player~$\Max$ if the projected play on $V$, 
$\langle v_0, v_1, v_2, \ldots \rangle$, is won by~$\Max$.  
For example, for the \emph{reachability} objective, the game states $(G, \own, \top)$ are the target states.  
The probability that player~$\Max$ wins an {\sgameone} on~$G$ 
is thus equal to the probability that~$\Max$ wins the corresponding stochastic game on~$\Gg$.

Directly solving the explicit game—namely, a stochastic reachability, parity, or energy game—yields decidability.  
However, this approach is inefficient, as the explicit game contains $3^{|V|} \cdot |V|$ states:  
there are $3^{|V|}$ possible ownership functions and $|V|$ possible token positions.  
Nevertheless, we show that the winning probability for player~$\Max$ can be decided in alternating polynomial time ($\mathbf{APTIME}$), which is equivalent to \PSpace\ \cite{CKS81},
matching the {\PSpace}-hard lower bound established in \Cref{subsec:pspace-hardness}.

We first observe that although the explicit game is large, it has a simple structure and directly inherits memoryless optimal strategies from the deterministic versions of the reachability, parity, and energy games.  
Furthermore, when both players follow their respective memoryless optimal strategies, any path leading to a winning cycle has polynomial length.  
Moreover, the winning probabilities for $\Max$ and $\Min$ at each game state can be represented in polynomial size.  
These properties together allow the winning probability at the initial state of the explicit game to be decided in $\mathbf{APTIME}$.



\begin{restatable}{lemma}{lemTossAsYouGo}\label{lem:toss-as-you-go-properties}
Our {\sgameone}s have
(a) optimal memoryless strategies. Moreover,
(b) the probability of winning can be written as a fraction with a denominator, which is the product of the denominators of the coin tosses of all unassigned states.
If the $\Max$ (resp. $\Min$) player plays optimally they can
(c) enforce that the first cycle%
\footnote{Note that, for reachability, we assume that target states have only target states as successors.}
reached and closed is winning with at least (resp. at most) their optimal winning probability,
(d) has a length of at most $|V|$, and
(e) the path to closing the cycles is of length at most $(|V|+1)\cdot (|V|+2)/2$.
\end{restatable}
\begin{proof}
To analyse the {\sgameone}s, we observe that it has a natural DAG structure, as the set of unassigned states can only fall.
We analyse the game inductively from the bottom SCC upwards, starting with the $2^{|V|}$ copies where the ownership of all nodes are assigned (induction basis).
For the length of the path until a cycle is closed, we show the stronger claim this to be (f) $\sum_{i=0}^m(|V|-i+1)$ if there are $m$ unassigned nodes. Note that (f) entails (e).

For this \textbf{induction basis}, the game is an ordinary non-stochastic two player game;
they are memoryless determined with memoryless optimal strategies  \cite{Ehrenfeucht1979,Gurvich88,mostowski1991games,EmersonJ91,BouyerFLMS08} (a,c), which guarantee to reach and close a winning cycle with length at most $|V|$ (d) in $|V|+1$ steps (f) for the winning player. 
Obviously, such strategies are also optimal.

As this only allows for the winning probabilities $0$ and $1$, they are fractions with a denominator being the empty product, showing (b).

For the \textbf{induction step}, we assume that the properties have been shown for up to $m$ unassigned states and move to $m+1$ unassigned states.

From an unassigned node $s$, the next step is a coin toss that assigns ownership.
The probability of winning is $\toss(s)$ times the probability of winning when starting from $s$ in the explicit game state that only differs from the current state in that $s$ is assigned to $\Max$ ($\Max$ winning the coin toss), plus $1-\toss(s)$ times the probability of winning when starting from $s$ in the explicit game state that only differs from the current state in that $s$ is assigned to $\Min$ ($\Min$ winning the coin toss).

By induction hypothesis, the probability can be written as a fraction with a denominator, which is the product of the denominators of the coin tosses of all unassigned states (b).

We also inherit that optimal memoryless policies of $\Max$ (resp. $\Min$) guarantee that the first cycle reached is winning with at least (at most) this probability (a).
The length of a path until closing this cycle is just one longer than the path from either successor state (c,d,f).

For states that are assigned an owner, we focus on the $|V|$ game states with a fixed partial assignment. For them, we design a game we call a strong (resp. weak) $p$-game for any $p\in [0,1]$ by turning all unassigned states with probability of winning $\geq p$ (resp. $>p$) into winning%
\footnote{For reachability games a winning (losing) sink is a final (non-final) state with only a self-loop, for parity games these sink states get an even (odd) priority, for energy games a positive (negative) weight.}
sinks, and those with probability of winning $<p$ (resp. $\leq p$) into losing sinks.

$\Max$ (resp. $\Min$) wins with probability at least $p$ (resp. $1-p$) if they win the strong (resp. weak) $p$-game by following their memoryless strategy in the strong (resp. weak) $p$-game on the game with the fixed partial assignment, and the optimal strategy that exists by induction hypothesis after leaving it. This is because, when following this strategy, either a winning cycle is closed within this area (which is then of length $\leq |V|-m$ is reached and closed in $\leq |V|-m+1$ steps, or a state from which there is an assignment is made by coin toss is reached within $|V|-m$ steps. From this state, $\Max$ (resp. $\Min$) wins with a probability of at least $p$ (resp. $1-p$) (see above), providing (d,f).

By their construction, for every state there is clear monotonicity in $p$-games, $\Max$ wins the strong $1$-game, $\Min$ wins the weak $0$-game, or there is a $p$ such that $\Max$ wins the strong and $\Min$ the weak $p$-game; in each of the three cases, this probability is the one each player can guarantee (a,c). In the latter case, this $p$ is a probability that arises for one of the unassigned states, providing (b) as shown above.

This completes the inductive argument and the proof.
\qed
\end{proof}

With \Cref{lem:toss-as-you-go-properties} at hand, we obtain an algorithm for deciding whether the winning probability meets a given threshold in $\mathbf{APTIME} = \mathbf{PSPACE}$.

\begin{theorem}\label{thm:pspace-toss-as-you-go}
Checking whether player $\Max$ can win the {\sgameone}\ with a given reachability, parity, or energy objective with a probability at least \modify{a given rational threshold~$\theta$, encoded in binary,} can be done in alternating polynomial time. 
\end{theorem}

\begin{proof}
We can simply play the explicit game forwards with an alternating machine, which first guesses the probability $p$ of winning, given by guessing the numerator and implicitly using the implicit denominator as the product of all denominators of the probability of all unassigned states (initially all states) given by $\mathtt{t}$.
We then check that $p \geq \theta$ holds and set the current node $v$ to $v_0$. 

We then repeat the following loop:
\begin{enumerate}
    \item  Write the current node on a string.
    \item  For states assigned to $\Max$ (resp.\ $\Min$), we existentially (resp.\ universally) choose a successor, update the current node $v$ to this successor.
    \item For unassigned nodes, existentially guess the numerator\footnote{When representing fractions by the numerator only, using the product of all denominators of unassigned nodes, then we can check $p \geq p_+\cdot n + p_- \cdot (d-n)$, where $d/n$ is the probability of Max winning the coin toss.}
    of the probability of winning for the case that $\Max$ wins the coin toss ($p_+$) and the chance of winning when $\Min$ wins the coin toss ($p_-$).
    We then universally choose whether to check that the weighted sum of these probabilities is at least the current probability of winning, or to universally make one of these assignments and update the probability accordingly to $p_+$ resp.\ $p_-$.
    \item We then offer both players in turn the option to claim that they have closed a winning cycle. 
    A player claiming this wins if they can and loses if they cannot show this. (Note that checking for having closed a winning cycle is an easy deterministic check.)
\end{enumerate}
If the actual probability is $p \geq \theta$, $\Max$ can always guess the correct probability and follow their optimal strategy.
If the actual probability is $p < \theta$, player $\Min$ can always follow their optimal strategy and either challenge in case there if there is a local inconsistency and otherwise challenge $p_+$ if $p_+$ is higher than the actual probability, and challenge $p_-$ otherwise.
\qed
\end{proof}

\subsection{PSPACE-hardness}\label{subsec:pspace-hardness}

To prove \PSpace-hardness of computing the winning probability in the \sgameone, we reduce from the satisfiability problem for quantified Boolean formulas (QBF) in prenex conjunctive normal form (PCNF).

For simplicity, we assume the quantifier prefix of the QBF is strictly alternating, so that the input has the form
\[
  \Phi \;=\; \forall x_1 \exists y_1 \; \forall x_2 \exists y_2 \; \dots \; \forall x_n \exists y_n \;\; \varphi,
\]
where the matrix $\varphi = C_1 \land C_2 \land \dots \land C_m$ is in CNF, and only contains variables occurring in the quantifier prefix.
It is well known that deciding satisfiability of such QBFs is \PSpace-complete~\cite{Papadimitriou1994}.
Moreover, their semantics can be understood in terms of the following (evaluation) game played by a universal and an existential player. The players assign variables following the order in the quantifier prefix, with the universal player assigning universally quantified variables $x_i$, and the existential player assigning existentially quantified variables $y_i$. The existential player wins if the resulting truth assignment satisfies the matrix, and the universal player wins if the assignment falsifies the matrix. A QBF is satisfiable (true) if the existential player has a winning strategy for the evaluation game, and unsatisfiable (false) if the universal player has a winning strategy.

To reduce from QBF satisfiability, we will construct a game graph $G_\Phi$ and show that the probability of the $\Max$ player winning is at least $\theta$ if $\Phi$ is satisfiable, and strictly less than $\theta$ if $\Phi$ is unsatisfiable, for a threshold value $\theta$ only depending on $n$.

\paragraph{Game graph construction.}
Given $\Phi$, we construct a directed game graph $G_\Phi$ with nodes as follows:  
for each \emph{universally quantified} variable $x_i$, nodes $\forall x_i$, $x_i$, $\neg x_i$;  
for each \emph{existentially quantified} variable $y_i$, nodes $\exists y_i$, $y_i$, $y_i'$, $\neg y_i$, $\neg y_i'$, $y_i''$, $\neg y_i''$;  
for each \emph{clause} $C_j$ of $\varphi$, a node $C_j$;  
one conjunction node $\land$;  
and two \emph{sink} nodes $\top$ and $\bot$.
The initial node is $\forall x_1$, and the target node is $\top$. 
Edges are added as follows:

\begin{description}
  \item[\textbf{Universal choice.}] From $\forall x_i$, add edges to $x_i$ and $\neg x_i$.
  \item[\textbf{Existential choice.}] From $\exists y_i$, add edges to $y_i''$ and $\neg y_i''$.
    From $y_i''$ add an edge to $y_i'$, and from $\neg y_i''$ an edge to $\neg y_i'$.
    Further, add edges from $y_i'$ to $y_i$ and from $\neg y_i'$ to $\neg y_i$.
  \item[\textbf{Quantifier progression.}] From $x_i$ and $\neg x_i$, add edges to $\exists y_i$.
    For $1 \le i < n$, add edges from $y_i$ and $\neg y_i$ to $\forall x_{i+1}$.
    For $i=n$, add edges from $y_n$ and $\neg y_n$ to $\land$.
  \item[\textbf{Clause choice.}] From $\land$, add an edge to each clause $C_j$.
  \item[\textbf{Literal choice.}] A literal is a variable or its negation.
    For each clause $C_j$ and each literal $\ell \in \{x_i,\neg x_i,y_i,\neg y_i\}$,
    if $\overline{\ell}\in C_j$ add an edge from $C_j$ to $\ell$ (so each clause connects to the \emph{negations} of its literals).
  \item[\textbf{Sink connections.}] Add edges to $\top$ from literal nodes $x_i,\neg x_i,y_i,\neg y_i$,
    from every $\forall x_i$, and from $\land$. Add edges to $\bot$ from each $\exists y_i$,
    the nodes $y_i',\neg y_i',y_i'',\neg y_i''$, and each $C_j$. Finally, add self-loops on $\top$ and $\bot$.
\end{description}

\begin{figure}[t]
\tikzset{
    state/.style={circle, draw, minimum size=8mm, inner sep=0pt},
    graystate/.style={state, fill=gray!40},
    every edge/.style={draw, ->, >=stealth}
}
    \centering
    \scalebox{0.9}{
    \begin{adjustbox}{trim=0pt 30pt 0pt 0pt, clip}
    \begin{tikzpicture}[node distance=.8cm and .25cm]
    
    \node[graystate,initial,initial text=] (Ax1) {$\forall x_1$};
    \node[graystate, above right=of Ax1] (x1) {$x_1$};
    \node[graystate, below right=of Ax1] (notx1) {$\neg x_1$};
    \node[state, right=1cm of Ax1] (Ey1) {$\exists y_1$};
    
    \node[state, above right=of Ey1] (y1pp) {$y''_1$};
    \node[state, right=of y1pp] (y1p) {$y'_1$};
    \node[graystate, right=of y1p] (y1) {$y_1$};
    
    \node[state, below right=of Ey1] (noty1pp) {$\neg y''_1$};
    \node[state, right=of noty1pp] (noty1p) {$\neg y'_1$};
    \node[graystate, right=of noty1p] (noty1) {$\neg y_1$};
    
    \node[graystate, below right=of y1] (Ax2) {$\forall x_2$};
    \node[graystate, above right=of Ax2] (x2) {$x_2$};
    \node[graystate, below right=of Ax2] (notx2) {$\neg x_2$};
    \node[state, right=1cm of Ax2] (Ey2) {$\exists y_2$};
    
    \node[state, above right=of Ey2] (y2pp) {$y''_2$};
    \node[state, right=of y2pp] (y2p) {$y'_2$};
    \node[graystate, right=of y2p] (y2) {$y_2$};
    
    \node[state, below right=of Ey2] (noty2pp) {$\neg y''_2$};
    \node[state, right=of noty2pp] (noty2p) {$\neg y'_2$};
    \node[graystate, right=of noty2p] (noty2) {$\neg y_2$};
    
    \node[graystate, below right=1cm of y2] (and) {$\land$};
    
    \node[state, below=1.1cm of noty1pp] (C1) {$C_1$};
    \node[state, right=2.2cm of C1] (C2) {$C_2$};
    \node[state, right=2.2cm of C2] (C3) {$C_3$};
    
    \path (Ax1) edge (x1);
    \path (Ax1) edge (notx1);
    
    \path (x1) edge (Ey1);
    \path (notx1) edge (Ey1);
    
    \path (Ey1) edge (y1pp);
    \path (y1pp) edge (y1p);
    \path (y1p) edge (y1);
    
    \path (Ey1) edge (noty1pp);
    \path (noty1pp) edge (noty1p);
    \path (noty1p) edge (noty1);
    
    \path (y1) edge (Ax2);
    \path (noty1) edge (Ax2);
    \path (Ax2) edge (x2);
    \path (Ax2) edge (notx2);
    
    \path (x2) edge (Ey2);
    \path (notx2) edge (Ey2);
    
    \path (Ey2) edge (y2pp);
    \path (y2pp) edge (y2p);
    \path (y2p) edge (y2);
    
    \path (Ey2) edge (noty2pp);
    \path (noty2pp) edge (noty2p);
    \path (noty2p) edge (noty2);
    
    \path (y2) edge (and);
    \path (noty2) edge (and);
    
    \path (and) edge [out=-90, in=-45] (C1);
    
    \path (and) edge [out=-90, in=-35] (C2);
    
    
    \path[->] (and) edge [bend left=40] node [right, near start,xshift=-1em,yshift=-1.2em]{} (C3);

    \path (C1) edge (notx1);
    \path (C1) edge (noty1);
    \path (C1) edge [bend right=17] (y2);
    
    \path (C2) edge (x2);
    \path (C2) edge (noty1);
    
    \path (C3) edge (notx2);
    \path (C3) edge (noty2);
    \end{tikzpicture}
    \end{adjustbox}
    }
\caption{Simplified game graph for the QBF in \Cref{example:qbfgraph}.
Shaded nodes have transitions to the target node $\top$, while all other nodes have transitions to $\bot$.}
\label{fig:qbfgraph}
\end{figure}
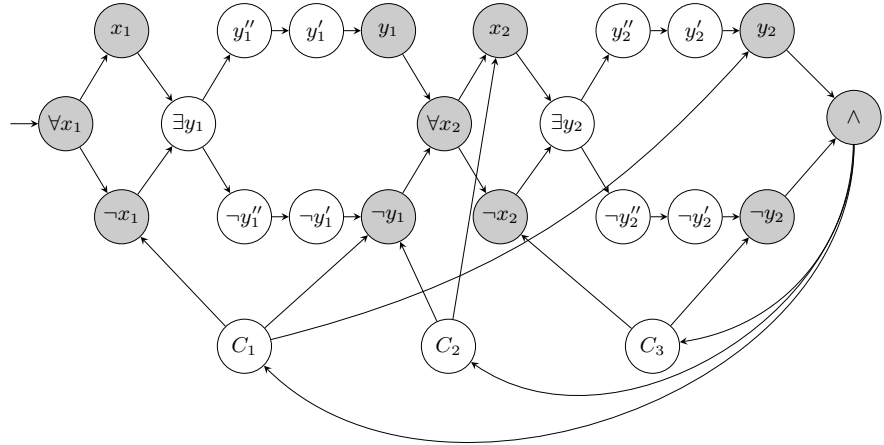
\begin{example}\label{example:qbfgraph}


Consider the QBF
\[
  \Phi \;=\; \forall x_1 \exists y_1 \forall x_2 \exists y_2 \; \varphi
\]
with matrix
\[
  \varphi \;=\; \underbrace{(x_1 \lor y_1 \lor \neg y_2 )}_{C_1} \;\land\; \underbrace{(\neg x_2 \lor y_1)}_{C_2} \;\land\; \underbrace{(x_2 \lor  y_2)}_{C_3}.
\]
\Cref{fig:qbfgraph} shows a simplified drawing of its game graph~$G_\Phi$ where the sinks and their incident edges are omitted. 
\qed
\end{example}

Consider the \sgameone\ with a reachability objective played on $G_\Phi$ where each coin toss is of probability $\frac{1}{2}$, that is, each node is assigned to player $\Max$ or $\Min$ with equal chance.
If a node with an edge to $\top$ is assigned to the $\Max$ player, $\Max$ can win immediately by following that edge. Similarly, we can assume that the $\Min$ player always transitions to $\bot$ when given the chance.
Conversely, it is clearly never in the interest of the $\Max$ player to go to $\bot$, or for the $\Min$ player to go to $\top$, so we can assume they never choose such a transition if any other transition is available.

The game only gets interesting once the token arrives at a clause node that is assigned to the $\Max$ player.
That only happens if, up to that point, any visited node with a $\top$-transition is assigned to $\Min$, and any visited node with a $\bot$-transition is assigned to $\Max$.
We now argue that the probability of this event is independent of players' choices up to that point (assuming they play rationally) and only depends on the number of variables in the quantifier prefix.

\begin{itemize}
    \item Let $p_i$ denote the probability of the $\Max$ player winning before the token reaches quantifier node $\forall x_i$ for $1 \leq i \leq n$, and let $p_{n+1}$ be the probability of them winning before the token gets to the $\land$-node.
    \item Further, let $q_i$ be the probability of the token arriving at node $\forall x_i$ for $1 \leq i \leq n$, and let $q_{n+1}$ be the probability of it reaching  the $\land$-node.
    \item Finally, let $q$ be the probability of the token reaching a clause node $C_j$ and the $\Max$ player gaining control, and let $p$ be the probability of the $\Max$ player winning before the token gets to a clause node $C_j$.
\end{itemize}
\begin{lemma}\label{lem:fixedprobabilities}
Given $n$, the probabilities defined above can be computed as follows:
\[
\begin{aligned}
q_1 &= 1, & q_{i+1} &= 2^{-6} q_i, & q &= 2^{-2} q_{n+1},\\
p_1 &= 0, & p_{i+1} &= p_i + q_i (2^{-1} + 2^{-2} + 2^{-6}), & p &= p_{n+1} + 2^{-1} q_{n+1}.
\end{aligned}
\]
\end{lemma}
\begin{proof}
Since $\forall x_1$ is the initial state, $q_1 = 1$ and $p_1 = 0$.
Reaching $\forall x_{i+1}$ for $1 \leq i < n$ requires reaching $\forall x_i$ with probability $q_i$, assigning $\forall x_i$, $(\neg) x_i$ and $(\neg) y_i$ to $\Min$, as well as assigning $\exists y_i$, $(\neg) y_i''$ and $(\neg) y_i'$ to $\Max$, so the probability is $q_{i+1} = 2^{-6} q_i$.
By the same argument, $q_{n+1} = 2^{-6} q_n$.
From node $\forall x_i$ for $1 \leq i \leq n$, the $\Max$ player can win before the token arrives at $\forall x_{i+1}$ (for $1 \leq i < n$) or $\land$ (for $i = n$) by having one of the nodes with edges to $\top$ assigned to them, which happens with probability $2^{-1} + 2^{-2} + 2^{-6}$.
Thus, $p_{i+1} = p_i + q_i (2^{-1} + 2^{-2} + 2^{-6})$.
From the $\land$-node, the $\Max$ player wins if that node is assigned to them, so $p = p_{n+1} + 2^{-1} q_{n+1}$, and the token is passed on to a clause node with $\Max$ in control with probability $\frac{1}{4}$, so $q = 2^{-2} q_{n+1}$.\qed
\end{proof}
The above lemma allows us to compute the probability $p$ of the $\Max$ player winning the game before gaining control of a clause node.
Next, we will show that they can win with probability $\frac{1}{2}$ from a clause node if they follow an existential winning strategy in the evaluation game for the QBF~$\Phi$.

\begin{lemma}\label{lem:existentialstrategy}
If the existential player has a winning strategy in the evaluation game for $\Phi$, then the $\Max$ player can win the \sgameone\ with a reachability objective on $G_\Phi$ with probability at least $p + \frac{1}{2}q$.
\end{lemma}
\begin{proof}
Interpreting a transition from $\forall v$ to $v$ as setting variable $v$ to true, and a transition to $\neg v$ as setting $v$ to false, before reaching a clause node, the $\Max$ player mimics the existential winning strategy for $\Phi$.
Assuming rational moves by the $\Min$ player, the $\Max$ player wins with probability $p$ before making it to a clause node $C_j$.
With probability $q$, they end up in a clause node $C_j$ that they control.
Since they mimicked an existential winning strategy, there is a literal $\ell \in C_j$ that is satisfied by the assignment induced by the players' moves. 
And so by construction of $G_\Phi$, there is an edge from $C_j$ to the negated literal node $\overline{\ell}$ that has not been visited before. The $\Max$ player moves to $\overline{\ell}$, and wins with probability $\frac{1}{2}$ by gaining control of this node and moving to $\top$.
Overall, the probability of the $\Max$ player winning is at least $p + \frac{1}{2}q$.\qed
\end{proof}
On the other hand, if the universal player has a winning strategy in the evaluation game, the $\Min$ player can mimic this strategy and force the $\Max$ player to visit a previously seen node, lowering their chances of reaching $\top$.
\begin{lemma}\label{lem:universalstrategy}
    If the universal player has a winning strategy in the evaluation game for $\Phi$, then the $\Max$ player can win the \sgameone\ with a reachability objective on $G_\Phi$ with probability at most $p + \frac{1}{4}q$.
\end{lemma}
\begin{proof}
    In this case, the strategy of the $\Min$ player is to mimic the universal winning strategy. As a consequence, upon gaining control of the $\land$-node, the $\Min$ player can transition to a clause $C_j$ that is falsified by the assignment induced by the players' choices. 
    Assuming rational moves, the $\Max$ player wins with probability $p$ before visiting a clause $C_j$, and arrives at $C_j$ and gains control with probability $q$. Since $C_j$ is falsified, and by construction of $G_\Phi$, there are only edges to literal-nodes $\ell$ that have been visited before and that are controlled by the $\Min$ player.
    Assume the $\Min$ player simply repeats their previous moves.
    The only chance for the $\Max$ player to win against this strategy is to eventually gain control of a literal node $y_i$ or $\neg y_i$ that has not been visited before and go to the target node $\top$ from there.
    To do that, they must first go through literal nodes $y_i''$ and $y_i'$ or $\neg y_i''$ and $\neg y_i'$, and gain control of both, which happens with probability $\frac{1}{4}$.
    Overall, the $\Max$ player can either win before they gain control of a clause node $C_j$ with probability $p$, or gain control of clause node $C_j$ with probability $q$, reach a previously unvisited literal node $y_i$ or $\neg y_i$ with probability at most $\frac{1}{4}$, and win from there.
    Thus the $\Max$ player can win with probability at most $p + \frac{1}{4}q$.\qed
\end{proof}


\begin{theorem}\label{thm:gameone-pspace-hard}
Deciding whether player $\Max$ can win an \sgameone\ with a reachability, parity, or energy objective 
with probability at least a given threshold $\theta$ is \PSpace-hard.
\end{theorem}
\begin{proof}
    Given a QBF~$\Phi$, we construct the game graph $G_\Phi$ and compute the probabilities $p$ and $q$ as defined above. Clearly, this can be done in polynomial time. 
    Let $\theta = p + \frac{1}{2}q$.
    
    It follows from \Cref{lem:universalstrategy} and \Cref{lem:existentialstrategy} that 
    player $\Max$ wins the \sgameone\ with a reachability objective on the game graph $G_\Phi$ with probability at least $\theta$ if, and only if, the QBF~$\Phi$ is satisfiable. 
    This establishes \PSpace-hardness for deciding whether the winning probability for player $\Max$ in an \sgameone\ with a reachability objective is at least a given threshold $\theta$. 

    This result can be easily extended to \sgameone s with parity or energy objectives. 
    In particular, an \sgameone\ with a reachability objective can be transformed into one with a parity or energy objective while preserving $\Max$'s winning probability. 
    For parity objective, one can add a self-loop to each target node, assign them an even priority of $2$, and assign all other nodes an odd priority of $1$. 
    For energy objectives, a self-loop can similarly be added to each target node with a weight of $1$, while all other nodes are assigned a weight of $-1$; the initial energy level can then be set to $(|V|+1)\cdot(|V|+2)/2$, which is the upper bound of a path to close a winning cycle according to~\Cref{lem:toss-as-you-go-properties}. \qed
\end{proof}

\section{\SGAMEtwo}\label{sec:toss-start}



We show that computing the exact winning probability of $\Max$ for \sgametwo, with any of the objectives—reachability, parity, or energy—is computationally hard.  
In fact, the problem is \SharpP-complete.
We then propose approximation algorithms that efficiently estimate this probability, and demonstrate empirically that they achieve fast convergence.

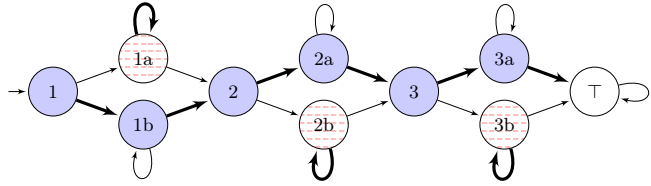
\begin{figure}[t] 
    \centering
    \scalebox{0.8}{
    \begin{tikzpicture}[scale=0.8, >=latex', shorten >=1pt, node distance=0.55cm and 1.5cm, on grid, auto]
    \node[state, initial, initial text=, fill=blue!20] (1) {1};
    \node[state, pattern={Lines[angle=0,distance=3pt,line width=0.3pt]}, pattern color=red!50, above right=of 1] (1a) {1a};
    \node[state, fill=blue!20, below right=of 1] (1b) {1b};
    
    \node[state, fill=blue!20, right=3cm of 1] (2) {2};
    \node[state, fill=blue!20, above right=of 2] (2a) {2a};
    \node[state, pattern={Lines[angle=0,distance=3pt,line width=0.3pt]}, pattern color=red!50, below right=of 2] (2b) {2b};
    
    \node[state, fill=blue!20, right=3cm of 2] (3) {3};
    \node[state, fill=blue!20, above right=of 3] (3a) {3a};
    \node[state, pattern={Lines[angle=0,distance=3pt,line width=0.3pt]}, pattern color=red!50, below right=of 3] (3b) {3b};
    
    \node[state, right=3cm of 3] (4) {$\top$};
    
    \path[->] (1) edge (1a);
    \path[->, line width=1.5pt] (1) edge (1b);
    \path[->] (1a) edge (2);
    \path[->,line width=1.5pt] (1b) edge (2);
    
    \path[->,line width=1.5pt] (2) edge (2a);
    \path[->] (2) edge (2b);
    \path[->,line width=1.5pt] (2a) edge (3);
    \path[->] (2b) edge (3);
    
    \path[->,line width=1.5pt] (3) edge (3a);
    \path[->] (3) edge (3b);
    \path[->,line width=1.5pt] (3a) edge (4);
    \path[->] (3b) edge (4);
    
    \path[->,line width=1.5pt] (1a) edge [loop above] ();
    \path[->] (2a) edge [loop above] ();
    \path[->] (3a) edge [loop above] ();
    
    \path[->] (1b) edge [loop below] ();
    \path[->,line width=1.5pt] (2b) edge [loop below] ();
    \path[->,line width=1.5pt] (3b) edge [loop below] ();
    
    \path[->] (4) edge [loop right] ();
    \end{tikzpicture}
    }
    \caption{A \sgametwo\ with a reachability objective in which player $\Max$ requires distinct strategies for exponentially many different arenas. 
    The figure illustrates one assignment for $\Gg_3$. 
    The optimal strategies for both players are highlighted using thick edges.
    }
    \label{fig:exponential}
\end{figure}

We begin with an example that sheds light on why computing the winning probability for player $\Max$ in \sgametwo\ is hard.
A strategy in \sgametwo\ of $\Max$, provides a strategy for each arena on $G$. 
Our objective is to compute the probability that player $\Max$ has a winning strategy in the randomly generated arena.
Note that it is possible that two arenas $\Aa$ and $\Aa'$ have the same strategy prescribed by $\Max$. 
However, for maximising winning probability, $\Max$ might require exponentially many different strategies. 
This is demonstrated by a simple class of games with reachability objectives, $\Gg_n$ which is a chain as shown in \Cref{fig:exponential} for $\Gg_3$. In this game, $\Max$ needs a specific path strategy which is following her own nodes, to win this particular game. And there are at least $2^n$ ways to assign the top and bottom nodes, such the path is unique, requiring a unique distinct strategy. 
This hints towards the fact that computing optimal strategy can be costly.

\subsection{\SharpP-Completeness}

The hardness follows from a reduction from a variant of the two-terminal reliability problem, which is known to be \SharpP-complete~\cite{Scott1986}.  
Given a directed graph $G=(V,E)$ with designated terminals $s,t \in V$, the two-terminal reliability problem asks for the probability that $s$ can reach $t$, assuming each edge $(u,v) \in E$ is independently present with probability $p$.  
This problem remains \SharpP-complete even for undirected or acyclic source-sink planar graphs of maximum degree three.  
We consider an \emph{adjusted} variant, also \SharpP-complete, where the start node $s$ has an incoming edge that exists with probability $p$; the probability of a path from $s$ to $t$ is thus $p \cdot \alpha$, where $\alpha$ is the reachability probability from $s$ to $t$ given that the incoming edge exists.  
To establish \SharpP-hardness for {\sgametwo}s, we give a logspace reduction from this adjusted problem to a {\sgametwo} with a reachability objective, and extend the argument to parity and energy objectives.
The intuition behind this reduction is that we construct a {\sgametwo} on a graph where the ownership of each node mimics the presence or absence of a corresponding edge in $G$.
 
For membership in \SharpP, the probability that $\Max$ wins a {\sgametwo} can be computed in \SharpP{} by enumerating all possible ownership assignments, checking for each whether $\Max$ has a winning strategy, and counting the fraction of assignments where $\Max$ wins.
The key difference between reachability objectives and parity or energy objectives is that, while a reachability game induced by a fixed ownership assignment can be solved in polynomial time, no such polynomial-time algorithm is known for parity or energy games. 
To obtain the \SharpP\ upper bound, we rely on the classic result that solving parity or energy games lies in \UP\ $\cap$ co-\UP~\cite{Jurdzinski98,BouyerFLMS08}, which allows to non-deterministically guess a unique short certificate for each game. 
Note that if the subproblem of deciding the winner in a parity or energy game were merely in \NP, this approach would fail: multiple short certificates could exist for a positive instance, so counting certificates would not correctly count winning assignments.


\begin{theorem}\label{thm:sharppcomplete}
Computing the winning probability of $\Max$ in \sgametwo\ with a reachability, parity, or energy objective is \SharpP-complete.
\end{theorem}


\subsection{FPRAAS}
\label{subsec:fpras}

Although computing the exact probability of winning in \sgametwo\ is \SharpP-complete, this does not rule out the possibility of efficiently approximating it. However, it was shown in \cite{provan1983complexity} that approximation of the two-terminal reliability problem within a given $\varepsilon>0$ additive error (specified as a rational number in the binary notation) is \SharpP-hard, so the same holds for approximating the value of our game. This still does not rule out a possibility of approximating this value with a high probability. Indeed, we show that such a scheme exists by leveraging a Monte Carlo method (see, e.g.\ \cite{Rubinstein2007}) to randomly sample two-player games from the given distribution and checking how often in them player $\Max$ wins the game. When this is done for ``long enough'', the frequency of winning the game in randomly chosen instances is a good approximation of the actual probability winning, but only with probability some less than 1. In order to formally define ``long enough'', we will make use of a special case of the Hoeffding's inequality when applied to the case where all the random variables are identically distributed and whose value can only be $0$ or $1$.

\begin{theorem}[Hoeffding's inequality~\cite{hoeffding1963probability}]
\label{thm:hoef}
Let $X_1, \ldots, X_n$ be independent and identically distributed random variables with $X_i \in \{0,1\}$ almost surely. Then, for any $\varepsilon > 0$,
$
\Pr\Big(\Big|\frac{1}{n}\sum_{i=1}^{n} X_i - \mathbb{E}[X_1]\Big| \ge \varepsilon \Big) \le 2 \, e^{-2 \varepsilon^2 n}.
$
\end{theorem}

\begin{definition}[Fully Polynomial Randomised Additive Approximation Scheme~(FPRAAS)]
Let $\mathcal{V}$ be a function mapping problem instances $x$ to a nonnegative real value $\mathcal{V}(x)$. 
A randomised algorithm $A$ is an additive \emph{$(\varepsilon,\delta)$-approximation scheme} for $\mathcal{V}$ if, 
for every instance $x$ and every $\varepsilon,\delta \in (0,1)$, the output $\hat{\mathcal{V}} = A(x,\varepsilon,\delta)$ satisfies
$
\Pr\big[|\hat{\mathcal{V}} - \mathcal{V}(x)| \le \varepsilon \big] \;\ge\; 1-\delta.
$

The algorithm $A$ is called a \emph{Fully Polynomial Randomised Additive Approximation Scheme (FPRAAS)}
if, for any input $(x,\varepsilon, \delta)$, $A$ is an additive $(\varepsilon,\delta)$-approximation scheme and its running time is polynomial in 
$|x|$, $\frac{1}{\varepsilon}$, and $\log \frac{1}{\delta}$.
\end{definition}

By Hoeffding's inequality (cf.~\Cref{thm:hoef}), choosing 
$
n \ge \frac{1}{2\varepsilon^2} \ln \frac{2}{\delta}
$ 
ensures that $A$ is indeed an additive $(\varepsilon,\delta)$-approximation scheme, and the running time is polynomial in $1/\varepsilon$ and $\ln(1/\delta)$.

Now, to get an FPRAAS each instance of the game on $G$ has to be solvable in polynomial time. Reachability games can be solved in linear time. Energy games can be solved in polynomial time when the  input weights are given in unary. Finally, parity games can be solved in polynomial time if the number of priorities is $\mathcal{O} (\log |G|)$ \cite{lehtinen2022recursive}.      

\begin{theorem}
{\Sgametwo}s with a reachability objective, a parity objective with $\mathcal{O}(\log |G|)$ priorities or an energy objective where weights are given in unary all admit an FPRAAS. 
\end{theorem}

\paragraph{\bf Implementation and Experiments.}

\begin{figure}[t]
    \centering
    \includegraphics[width=.8\linewidth]{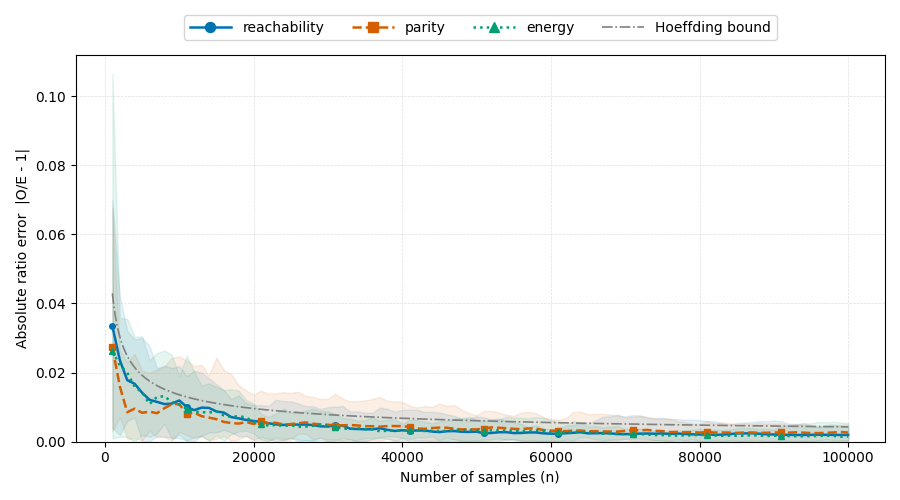}
    \caption{
    Error convergence of the FPRAAS for {\sgametwo} with reachability, parity, and energy objectives.
    Each line shows the mean absolute ratio error across 10 random game graphs, with shaded regions indicating variance.
    }
    \label{fig:fpras}
\end{figure}

We implemented the FPRAAS for {\sgametwo} with reachability, parity, and energy objectives~\cite{artifact2026}. 
The associated artifact includes the complete source code, experiment scripts, and datasets necessary to reproduce the results in this paper. 
It builds on two external tools: \href{https://github.com/gamegraphgym/ggg}{ggg}~\cite{ggg} for reachability and parity games, whose solvers are based on~\cite{BrimCDGR11,BenerecettiDM18}, and \href{https://github.com/pazz/egsolver}{egsolver}~\cite{egsolver} for energy games, whose solver is based on~\cite{BrimCDGR11}.
All experiments were performed on a machine with an Apple M3 8-core CPU and 16 GB of RAM running macOS. 
For each objective, we generated 10 random game graphs with 20 nodes and maximum outdegree of 20; parity objectives used up to 6 priorities, and energy objectives used an initial credit of 0 with weights in $\{-10, \ldots, 10\}$.
We computed the exact winning probability for $\Max$ by exhaustively solving all $2^{20}$ ownership assignments and compared it to the approximate probability from 100,000 random samples.
The mean absolute ratio error, $|O/E - 1|$, was computed across all 10 graphs.
\Cref{fig:fpras} shows the error (y-axis) against accumulated samples~(x-axis), with blue solid, orange dashed, and green dotted lines for reachability, parity, and energy, respectively.
Empirical errors converge much faster than Hoeffding’s bound ($\varepsilon=0.005$, $\delta=0.05$, dash-dotted grey line), and the variance decreases with increasing samples, as reflected by the narrowing shaded regions.
While exact computation takes over an hour for reachability and parity and over 24 hours for energy, our FPRAAS achieves comparable accuracy in about 6 minutes for reachability and parity, and 2 hours for energy, demonstrating that the FPRAAS is both efficient and effective.

\section{Discussion}\label{sec:discussion}
We studied the computational complexity of two-player infinite-duration games in which control of nodes is assigned randomly—either dynamically upon first visit or statically before play begins. These settings differ from the classical random-turn games, where control may be reassigned each time a node is visited.
Our results establish a clear complexity landscape across these variants and objectives. 
While most qualitative questions are \NL-complete, quantitative problems exhibit higher complexity: \PSpace-complete for \sgameone{} and \SharpP-complete for \sgametwo{}. 
To complement these hardness results, we developed efficient approximation schemes for the latter case and demonstrated their practical effectiveness.

We also explored the possibility of efficiently reducing stochastic games to random-turn games; although our findings provide partial insights, the existence of such a reduction remains an open question. 
A promising direction for future work is to investigate intermediate mechanisms between random-turn games and \sgameone{}, where ownership is fixed only on the $k$-th visit. 
Studying how increasing $k$ influences the winning probabilities could reveal a more nuanced understanding of the relationship between randomness in control and strategic outcomes.

Beyond theoretical interest, random arena games have potential applications in game design, particularly in ensuring fairness in competition. 
Our methods for computing exact winning probabilities can be used to assess whether a given game design between two competitive players is fair—i.e., whether both players have nearly equal chances of winning. 
Moreover, by adjusting node control probabilities or fixing control on selected nodes, one could systematically balance strategic advantages and improve fairness in game mechanics.

\begin{credits}
\subsubsection{\ackname} This work was supported by the EPSRC through grants EP/X03688X/1 and EP/X042596/1.

\end{credits}
%
%
%
\newpage
\bibliographystyle{splncs04}
\bibliography{main}

\newpage
\appendix

\section{Missing Proofs for \Cref{sec:qualitative}}\label{app:qualitative}

Given a random-turn game with a reachability objective played on $G = (V, E)$, we can reduce it to a stochastic reachability game played on $G' = (V', E')$.
We have $\top \in V'$ and, 
for each node $v \neq \top$ in $G$, we create three copies in $G'$, making $v$ in $G'$ a random node and $(v^+)$ and $(v^-)$ owned by $\Max$ and $\Min$ respectively. 
The set of edges $E'$ is defined as follows:
\begin{itemize}
    \item $(v, v^+), (v, v^-) \in E'$ each transition with probability $\toss(v)$ and $1 - \toss(v)$, respectively, for all random nodes $v \in V'$;
    \item $(v^+, u), (v^-, u) \in E'$ for all $\Max$-controlled nodes $v^+$ and $\Min$-controlled nodes $v^-$ such that $(v, u) \in E$. 
\end{itemize}
 
 Consider the game graph in \Cref{fig:explicit-game}(a), we have the stochastic game reduced from it shown in the top figure of \Cref{fig:explicit-game}(b). 

 \thmAlmostRTG*
\begin{proof}
We first show that deciding whether $\Max$ wins almost surely is in~\NL.
Given a random-turn game played on a game graph $G$, let $\toss$ denote the coin-toss function.
Observe that a node in $G$ is almost-surely (and surely) losing for $\Max$ if, and only if, it cannot reach the target state.
For the ``only if" direction, the statement is immediate: if a node cannot reach the target, then $\Max$ loses surely and thus almost surely.
For the converse, suppose that a node $v$ can reach the target.
Let $v_0v_1v_2\ldots v_n$ be the shortest path from $v=v_0$ to the target $v_n$, and let $\Max$ play to follow this path.
Then the probability of reaching the target from $v$ is at least $\prod_{i=0}^{n-1}\toss(v_i) > 0$, implying that $v$ is not almost-surely losing for $\Max$.

Next, we show that a node $v$ is almost-surely winning for $\Max$ if, and only if, it cannot reach a surely losing node.
It is clear that reaching a surely losing node makes the overall winning probability less than one.
Conversely, suppose $v$ cannot reach a surely losing node; equivalently, every state reachable from $v$ has a path to the target.
For each such node, we fix $\Max$’s strategy to follow a shortest path to the target.
The stochastic game induced by these fixed strategies (for $\Max$) and any strategy of $\Min$ then becomes a Markov chain in which the target node is the only bottom SCC.
Hence, starting from $v$, $\Max$ wins almost surely in both the induced stochastic game and the original random-turn game.

To decide the complement problem, whether $\Max$ does not win almost surely, we have an {\NL} algorithm that searches for a path from the initial node to a node $v$ that cannot reach the target.
This can be done in \NL, since both the st-connectivity and the st-non-connectivity problems are \NL-complete.
Since $\mathbf{NL} = \mathbf{co}\mbox{-}\mathbf{NL}$, the problem whether $\Max$ wins almost surely is in \NL.

For \NL-hardness, we reduce from the st-connectivity problem, which asks whether $t$ is reachable from $s$ in a directed graph $G$.
Given an instance $\langle G, s, t \rangle$, we construct a random-turn game on a graph $G'$ by having $s$ as the initial node and $t$ as the target node.
Moreover, for every other node $v$ in $G$, we add an edge $(v, s)$ in $G'$.
Then $\Max$ wins almost surely on $G'$ if and only if $t$ is reachable from $s$ in $G$.

If $t$ is reachable from $s$, then every node in $G'$ can reach $s$, and no node is surely losing for $\Max$.
Thus, $\Max$ wins almost surely from $s$.
Otherwise, if $t$ is not reachable from $s$ in $G$, it remains unreachable in $G'$, making $s$ a surely losing node for $\Max$ with a winning probability of $0$.

Therefore, deciding whether $\Max$ wins almost surely in a random-turn game is \NL-complete.
\qed
\end{proof}

\thmSureRTG*
\begin{proof}
We show that $\Max$ does not win surely if and only if there is a ``bad'' lasso-path in the game graph which does not contain the target state. 
Guessing such a lasso path can be done in \NL.
For \NL-harness, we reduce from the st-non-connectivity problem asking whether $t$ is \emph{not} reachable from $s$ on a directed graph $G$, which is \NL-complete.
We construct a game graph $G'$ with $n$ layers of nodes, where each node $(v,i)$ representing node $v$ can be reached from $s$ after exactly $i$ steps in $G$.
We also add a fresh target node $\top$ and a sink node $\bot$, which is only reached from $(t, i)$ for all $0\le i \le n$.
The lasso paths in $G'$ should contain the sink node $\bot$ since the only cycle in $G'$ is the self-loop on $\bot$. 
Now, it is easy to see that $t$ is not reachable from $s$ in $G$ if and only if $\Max$ wins the random-turn game on $G'$, that is, there is no lasso path in $G'$.
\qed
\end{proof}

\thmSureAllGames*
\begin{proof}
We first observe that sure winning for {\sgameone} and {\sgametwo} is the same as sure winning for random-turn games.
Since the existence of a ``bad'' lasso path witnesses that $\Max$ cannot surely win any of the three games.
It follows from \Cref{thm:sure-rtg} that it is \NL-complete to decide whether $\Max$ surely wins an {\sgameone} or a {\sgametwo} with a reachability, parity or energy objective.

The next thing we observe is that almost-sure winning coincides with sure winning for {\sgameone}s and {\sgametwo}s.    
It is obvious that if $\Max$ wins a game surely, then they win it almost surely.
However, in general, it is not necessary that $\Max$ wins a game surely if they win it almost surely, as we have seen for the random-turn games.
For {\sgameone}s and {\sgametwo}s, we give a contrapositive argument:
if $\Max$ does not win the game surely, $\Max$ does not win the game almost surely.
Assume $\Max$ does not win the game surely, then there is a ``bad'' lasso path in the game graph.
For {\sgameone}s, a game state in which every node on the ``bad'' lasso-path is assigned to $\Min$ is generated with positive probability, while for {\sgametwo}s, if player $\Min$ follows this lasso-path, $\Min$ wins with positive probability since this game state is reachable with positive probability. In both cases, the winning probability of $\Max$ is strictly less than one.
\qed
\end{proof}

\section{Missing Proofs for \Cref{sec:toss-start}}\label{app:toss-at-start}

\subsection{\SharpP-Hardness} 
\label{subsec:sharpp-lower}

We show that the problem is \SharpP-hard by reducing from two-terminal reliability problem. 
Given a finite directed graph \(G=(V,E)\) with $E \subseteq V \times V$, 
the two-terminal reliability problem asks for the probability that a designated terminal \(s\) can reach a target terminal \(t\), 
assuming each edge \(e\in E\) is present independently with probability \(p\).
It is known that solving two-terminal reliability problem is \SharpP-complete, even for undirected and acyclic directed source-sink planar graphs having degree at most three, that is, a node has at most three outgoing edges~\cite{Scott1986}. 


We consider a variant of the two-terminal reliability problem, called the \emph{adjusted} two-terminal reliability problem, 
in which we assume that the start node \(s\) has an incoming edge, and this incoming edge is present with probability \(p\).
Therefore, the probability of having a path from $s$ to $t$ in the variant would be $p \times \alpha$, where $\alpha$ is the probability of reaching $t$ from $s$ considering that the incoming edge to $s$ exists. Hence, the following result holds:

\begin{corollary}\label{cor:adjusted-reliability}
	The \emph{adjusted} two-terminal reliability problem remains \SharpP-complete.
\end{corollary}

To prove \SharpP-hardness of {\sgametwo}s with any of the objectives, we first show that there is a logspace reduction from \emph{adjusted} two-terminal reliability problem to \sgametwo\ with a reachability objective.
We then show this hardness result can be easily extended to parity and energy objectives.


\paragraph{Game graph construction.} 
Given an instance of the \emph{adjusted} two-terminal reliability problem $\langle G,s,t,p  \rangle$, we construct a directed game graph $G' = (V', E')$ as follows:
For every edge $(u, v) \in E$ such that $(u,v)$ is not the incoming edge of $s$, we introduce a node $e_{uv}$. Moreover, we introduce three new nodes $init$, $\top$ and $\bot$. Hence, $V' = \left\{e_{uv} \mid (u,v) \in E\right\} \cup \{init, \bot, \top\}$ and $|V'| = |E| + 2$.
Edges are added as follows: 
\begin{itemize}
	\item
	For every node $e_{uv} \in V'$ such that $u=s$, add edge $(init, e_{uv})$.
	\item
	For every node pair $e_{uv}, e_{vw} \in V'$, add edge $(e_{uv}, e_{vw})$.
	\item
	For every node $e_{uv} \in V'$ such that $v=t$, add edge $(e_{uv}, \top)$.
	\item 
	For every node $v' \in V'$ such that $v' \neq \top$, add edge $(v', \bot)$.
    \item 
    For $\bot$, add self-loops, that is, edges $(\bot,\bot)$ and $(\top,\top)$.
\end{itemize}
In the constructed game graph $G' = (V', E')$, each edge $(u,v) \in E$ in $G$ (except the incoming edge to $s$) is represented by an intermediate node $e_{uv} \in V'$, whose ownership is determined randomly: with probability $p$ it is assigned to $\Max$ (indicating the edge is kept), and with probability $1 - p$ it is assigned to $\Min$ (indicating the edge is removed). 
That is, we have the coin toss function $\toss(v) = p$ for all $v \in V'$. 
The incoming edge to $s$ is represented by the node $init$, which is also assigned randomly: if $\Max$ controls $init$, the edge exists and the game can proceed; otherwise, the game moves to $\bot$.

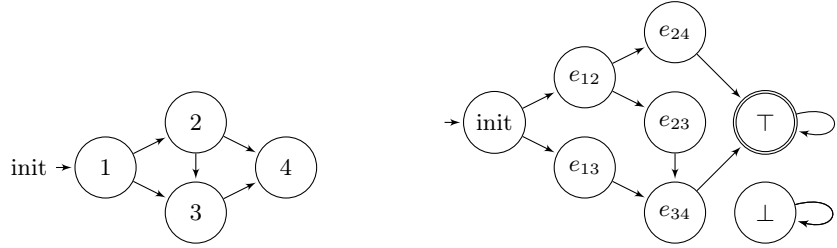
\begin{figure}[t]
    \centering
    \begin{subfigure}{.48\textwidth} 
    \centering
    \begin{tikzpicture}[scale=0.6, >=latex',shorten >=1pt,node distance=3cm,on grid,auto]
    \begin{scope}[shift={(0,5)}]   

    \node[state,initial, initial text=init] (1) at (0,0) {1};
    \node[state] (2) at (2,1) {2};
    \node[state] (3) at (2,-1){3};
    \node[state] (4) at (4,0) {4};
    
    \path[->] (1) edge node [midway, above] {} (2);
    \path[->] (1) edge node [midway, above] {} (3);
    \path[->] (2) edge node [midway, above] {} (3);
    \path[->] (2) edge node [midway, above] {} (4);
    \path[->] (3) edge node [midway, above] {} (4);
    \end{scope}

    \end{tikzpicture}
    \caption{A game graph $G$ for an instance of the \emph{adjusted} two-terminal reliability problem with start node $1$ and target node $4$.}
    \label{fig:reduction-a}
    \end{subfigure}
    \hfill
    \begin{subfigure}{.48\textwidth} 
    \centering
    \begin{tikzpicture}[scale=0.6, >=latex',shorten >=1pt,node distance=3cm,on grid,auto]
    \node[state] (s1) at (6, -2) {$\bot$}; 
    \node[initial, state ,initial text={}] (init) at (0, 0) {init}; 
    
    \node[state] (e12) at (2,1) {$e_{12}$};
    \node[state] (e13) at (2,-1){$e_{13}$};
    
    \node[state] (e24) at (4,2) {$e_{24}$};
    \node[state] (e23) at (4,0) {$e_{23}$};
    
    \node[state] (e34) at (4,-2) {$e_{34}$};
    
    \node[state, accepting] (T) at (6,0) {$\top$}; 
    
    \path[->] (init) edge (e12)
          (init) edge (e13);
    
    \path (s1) edge [loop right] (s1); 
    
    \path[->] (e12) edge (e24)
          (e12) edge (e23);
    
    \path[->] (e13) edge (e34);
    
    \path[->] (e24) edge (T);
    
    \path[->] (e23) edge (e34); 
    
    \path[->] (e34) edge (T);
    
    \path (T) edge [loop right] (T); 
    \path (s1) edge [loop right] (s1); 
    \end{tikzpicture}
    \caption{The game graph $G'$ for \sgametwo\ constructed from $G$ on the left.}
    \label{fig:reduction-b}
    \end{subfigure}
    \caption{Game graph construction for the \SharpP-hardness reduction. 
    In the game graph $G'$ on the right, $init$ is the start node, $\top$ is the target, and all nodes except $\top$ have an directed edge to $\bot$ which is omitted in the figure for simplicity. 
    }
    \label{fig:hardness-reduction}
\end{figure}

An example of the \SharpP-hardness reduction is provided in \Cref{fig:hardness-reduction}: 
given a game graph $G$ representing an instance of the \emph{adjusted} two-terminal reliability problem in \Cref{fig:reduction-a}, a game graph $G'$ for \sgametwo\ in \Cref{fig:reduction-b} is constructed.

Intuitively, the presence of an edge \((u,v) \in E\) (resp. \(init\)) in \(G\) is reflected by the ownership of the corresponding node \(e_{uv}\) (resp. \(init\)) in \(G'\): 
if the edge is present in \(G\), then player \(\Max\) owns the corresponding node in \(G'\) and can choose to move to a successor other than the rejecting sink \(\bot\) to continue the reachability game; 
otherwise, player \(\Min\) owns this node and transition to the sink \(\bot\) to win the reachability game.

\begin{lemma}\label{lem:reliability}
    Let $\langle G,s,t,p  \rangle$ be an instance of the \emph{adjusted} two-terminal reliability problem. 
    Let $\langle G', \toss \rangle$ be an instance of the \sgametwo\ with a reachability objective where $G' = (V', E')$ is the game graph constructed as above and $\toss(v) = p$ for all $v\in V'$.
    We have the probability that $\Max$ wins the game on $G'$ is equal to the probability of the \emph{adjusted} two-terminal reliability problem $\langle G,s,t,p  \rangle$.
\end{lemma}

\begin{proof}
    Given an instance of the \emph{adjusted} two-terminal reliability problem $\langle G,s,t,p  \rangle$, we construct a \sgametwo\ $\langle G', \toss \rangle$ with a reachability objective as described above.

    Given a function $\chi: E \to \{0, 1\}$, where $\chi(e) = 1$ indicates that $e \in E$ is present in $G$ and $\chi(e) = 0$ otherwise, a subgraph $G(\chi)$ of $G$ is induced.  
    There are in total $2^{|E|}$ such functions, and each corresponding subgraph is generated with equal probability.  

    We can then define an ownership assignment function $\own_{\chi}: V'\setminus\{\bot,\top\} \to \{\Max, \Min\}$ on $G'$ such that $\own_{\chi}(init) = \Max$ if $\chi(init)=1$, $\own_{\chi}(e_{uv}) = \Max$ if $\chi((u,v)) = 1$ for $(u,v) \in E$,  and $\own_{\chi}(v) = \Min$ otherwise.  

    There are in total $2^{|E|}$ such assignment functions, and each assignment is generated with equal probability.  
    Similarly, given an ownership assignment function $\own_{\chi}$ on $G'$, we can define a corresponding function for $G$ indicating the presence of the edges.  

    It is easy to see that $s$ can reach $t$ in $G(\chi)$ if and only if $\Max$ wins the game on $G'$ under the ownership assignment function $\own_{\chi}$:  
    a path $v_0, v_1, \ldots, v_n$ in $G(\chi)$, where $v_0 = s$, $v_n = t$, and each edge $(v_i, v_{i+1})$ for $i \in [0,n-1]$ is present under $\chi$, 
    corresponds to the path $init, e_{sv_1}, \ldots, e_{v_{n-1}t}, \top$ in $G'$, where all these nodes are owned by $\Max$ under $\own_{\chi}$.

    Therefore, the probability that $\Max$ wins a {\sgametwo} on $G'$ is exactly the probability of the adjusted two-terminal reliability problem on $G$.\qed
\end{proof}

Since the adjusted two-terminal reliability problem is \SharpP-complete by \Cref{cor:adjusted-reliability}, and we have shown a logspace reduction from the adjusted two-terminal reliability problem to \sgametwo\ with a reachability objective by \Cref{lem:reliability}, we have established \SharpP-hardness for \sgametwo\ with a reachability objective.

	
	

Similar to \Cref{thm:gameone-pspace-hard}, this hardness result can be easily extended to {\sgametwo}s with parity or energy objectives. 
In particular, a \sgametwo\ with a reachability objective can be transformed into one with a parity or energy objective while preserving $\Max$'s winning probability. 
For a parity objective, one can add a self-loop to each target state, assign them an even priority of $2$, and assign all other states an odd priority of $1$. 
For an energy objective, a self-loop can similarly be added to each target state with a weight of $1$, while all other states are assigned a weight of $-1$; the initial energy level can then be set to $|V|+1$, the upper bound of the length of a path to close a winning cycle. 

It is not hard to see that $\Max$ wins a reachability game on $G$ if and only if $\Max$ wins the parity game (resp. energy game) on the new game graph $G'$ under the same random ownership assignment.
We have:
\begin{theorem}\label{thm:sharpphard-parity}
	Computing the winning probability of $\Max$ in a {\sgametwo} with a reachability, parity or an energy objective is \SharpP-hard.
\end{theorem}

\subsection{Membership in \SharpP}
\label{subsec:sharpp-upper}

We now show membership in \SharpP.
\begin{theorem}\label{thm:sharpp-upperbound}
	Computing the winning probability of $\Max$ in a {\sgametwo} with a reachability, parity or an energy objective is in \SharpP.
\end{theorem}
\begin{proof}
    Let $\langle G, \toss \rangle$ be an instance of {\sgametwo}, where $G = (V, E)$ is a game graph and $\toss: V \to [0,1]$ is the coin-toss function.
For each $v \in V$ with $\toss(v) \in (0,1)$, represent $\toss(v)$ as a reduced fraction $a_v / b_v$ with positive coprime integers $a_v$ and $b_v$ encoded in binary.
We encode each possible ownership assignment by a string 
\modify{$str$ consisting of $|V|$ parts, where each part $i$ ($1 \leq i \leq |V|$) corresponds to a node $v_i \in V$, and the part $i$ of $str$ determines ownership as follows:
\begin{itemize}
\item if the $i$-th part of $str$ encodes a value $< a_{v_i}$, then $v_i$ is assigned to $\Max$;
\item if the $i$-th part of $str$ encodes a value $\ge a_{v_i}$ and $< b_{v_i}$ it is assigned to $\Min$;
\item otherwise, it is \emph{not} a legal encoding and should be ignored. 
\end{itemize}
This encoding is polynomially bounded, as the length of each string is bounded by $\max_{v \in V}\lceil \log(b_v)\rceil \cdot |V|$.
}

Each assignment $str$ induces a standard two-player game with a reachability, parity, or energy objective.
For the reachability objective, deciding whether $\Max$ has a winning strategy in the induced game is in $\mathbf{P}$,
so counting the number of ownership assignments where $\Max$ wins is in~\SharpP.

For parity and energy objectives, each induced game admits a unique short (polynomial-size) certificate witnessing $\Max$’s victory,
as deciding the winner in parity or energy games is in~\UP~\cite{Jurdzinski98,BouyerFLMS08}.
Using the nondeterministic power of a \SharpP\ machine, we can guess and verify this unique certificate for each ownership assignment.
Thus, the number of accepting computation branches equals the number of ownership assignments in which $\Max$ wins.

Therefore, computing the number of ownership assignments (and hence the probability) in which $\Max$ has a winning strategy in a {\sgametwo} with a reachability, parity, or energy objective is in~\SharpP.
\qed
\end{proof}

\end{document}